\documentclass[12pt]{article}
 
\usepackage{amsfonts, amsthm}
\usepackage{amssymb}
\usepackage{amsmath}  
\usepackage{graphicx, comment, xcolor, hyperref}
\usepackage{authblk}
\usepackage{mathrsfs}

\usepackage[a4paper]{geometry}
\newtheorem{theorem}{Theorem}

\newtheorem{lemma}[theorem]{Lemma}

\newtheorem{definition}[theorem]{Definition}
\theoremstyle{remark}

\newcommand{\ket}[1]{|#1\rangle}
\newcommand{\bra}[1]{\langle#1|}
\newcommand{\ketbra}[2]{|#1\rangle\langle#2|}

\newcommand{\tr}{\text{\rm tr}}
\newcommand{\ttr}[1]{\text{\rm tr}\left(#1\right)}

\newcommand{\cC}{\mathcal{C}}

\newcommand{\cN}{\mathcal{N}}
\newcommand{\cK}{\mathcal{K}}

\newcommand{\cL}{\mathcal{L}}

\newcommand{\eps}{\varepsilon}
\newcommand{\cW}{\mathcal{W}}

\newcommand{\Var}{\mathsf {Var}}
\newcommand{\Cov}{\mathsf{Cov}}
\newcommand{\E}{\mathsf{E}}

\title{Covariance Decomposition as a Universal Limit on Correlations in Networks}

\author{Salman Beigi$^1$, Marc-Olivier Renou$^2$}
\affil{\it \footnotesize $^1$School of Mathematics, Institute for Research in Fundamental Sciences (IPM)\\ \it \footnotesize P.O.~Box 19395-5746, Tehran, Iran\\
\it \footnotesize  $^2$ICFO-Institut de Ciencies Fotoniques, 08860 Castelldefels (Barcelona), Spain}

\begin{document}

\maketitle

\begin{abstract}
Parties connected to independent sources through a network can generate correlations among themselves. 
Notably, the space of feasible correlations for a given network, depends on the physical nature of the sources and the measurements performed by the parties.
In particular, quantum sources give access to nonlocal correlations that cannot be generated classically. 
In this paper, we derive a universal limit on correlations in networks in terms of their covariance matrix. 
We show that in a network satisfying a certain condition, the covariance matrix of any feasible correlation can be decomposed as a summation of positive semidefinite matrices each of whose terms corresponds to a source in the network.
Our result is universal in the sense that it holds in \emph{any} physical theory of correlation in networks, including the classical, quantum and all generalized probabilistic theories.

\end{abstract}


\section{Introduction}

The causal inference problem, i.e., determining the causes of observed phenomena, is a fundamental problem in many disciplines. 
Given two categories of \emph{observed} and \emph{latent} variables, it asks to deduce the causation patterns between the two sets of variables, where the former ones can be measured while the latter ones are hidden.  More specifically, the problem is whether the statistics of the observed variables can be modeled with a given causal pattern~\cite{Spirtes2000CausationPrediction,Pearl2009Causality}. 

In this paper, we restrict ourselves to \emph{network} causation patterns, which are determined by a bipartite graph with one part for observed variables and one part for latent variables~\cite{Fritz2012BeyondBell,Branciard2010}. We can think of the vertices associated to latent variables as \emph{sources} that produce \emph{signals} that are \emph{independent} of each other. We can also think of vertices associated to observed variables as \emph{parties} who receive signals from their adjacent sources. After receiving these signals, each party performs a measurement whose output determines the value of the corresponding observed variable. For instance, if the signals are classical random variables, the measurement is a function of these variables. Here, we emphasize that the distribution of signals (latent variables) are mutually independent, and an observed variable equals the measurement outcome applied on the adjacent sources. See Figure~\ref{fig:nosignaling} for an example of a network. Now given such a network with $n$ parties, and  a joint probability distribution $p(a_1, \dots, a_n)$ on their observed variables, the question is, can $p(a_1, \dots, a_n)$ be modeled as the output distribution of the measurement of sources distributed in the network? 

Quantum physics introduces seminal aspects to the causal inference problem~\cite{HensonLimitCorrelationGeneralisedBayNet,FritzBeyondBell2,Allen2017QcommonCauses}. The point is that the sources may produce quantum signals, i.e., multipartite entangled states, and parties may perform quantum measurements on the received signals. In this case the answer to the causal inference problem may differ from its answer in the completely classical (local) setting. Indeed, already in the primary example of Bell's scenario we observe nonlocal correlations in the quantum setting that cannot be modeled classically~\cite{bell1964einstein}; this is Bell's nonlocality. It is the elementary manifestation of a broader phenomenon, called \emph{network nonlocality}, in which quantum nonlocal correlations admitting a model in some network cannot be modeled classically in the same network. 

The existence of quantum models, which give access to more correlations than classical ones, raises the question of the ultimate limits on correlations in networks.
In the case of Bell's scenario, this ultimate limit is given by the \emph{no-signaling principle}~\cite{popescu1998causality}. 
We imagine that any correlation satisfying this principle can be obtained through `measurements' of `states' in some hypothetical \emph{Generalised Probabilistic Theory} (GPT)~\cite{Barrett2007GPT, Short2010couplers, Janotta2013norestriction, Skrzypczyk2009couplers}. Such theories generalize quantum physics and formalize the Popescu-Rohrlich box~\cite{Popescu1994,Scarani2006PRbox}.
However, the no-signaling principle alone is insufficient to define the ultimate limits of correlation in networks~\cite{gisin2020constraints, bancal2021}. 
In the following, we define these limits through the \emph{causality principle} (see Figure~\ref{fig:nosignaling}).
This resolusion, combined with the \emph{inflation technique}~\cite{wolfe2019inflation}, allow us to derive a universal limit on correlations in networks.

\begin{figure}
\begin{center}
\includegraphics[scale=0.35]{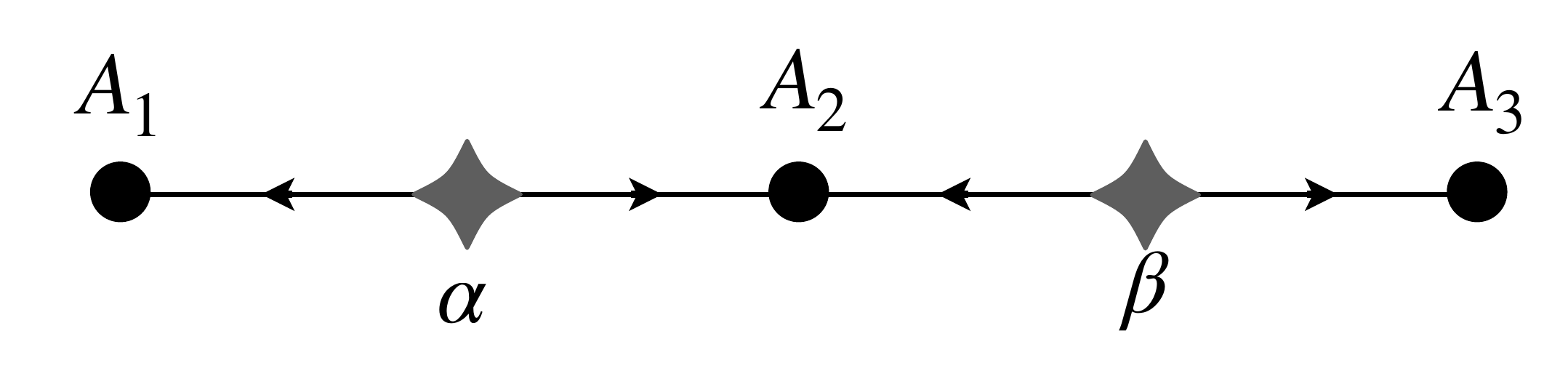}
\caption{\footnotesize This is an example of a network with two sources $S_\alpha, S_\beta$ and three parties $A_1, A_2, A_3$. Each party after receiving the signal from its adjacent sources, performs a measurement to determine her output. We denote the output of party $A_i$ by $a_i$, and the joint distribution of outputs by $p(a_1, a_2, a_3)$.\newline
\emph{First causality law:} Imaging that this network is a part of a larger network, yet the structure of the network in the neighborhood of $A_2$ is the same. Then the marginal distribution $p(a_2)$ in the larger network remains the same.\newline
\emph{Second causality law:}
Since in this network $A_1$ and $A_3$ do not have a common source, we have that $p(a_1, a_3) = p(a_1) p(a_3)$. 
}
\label{fig:nosignaling}
\end{center}
\end{figure}  

Recently, Kale et al~\cite{kela2019semidefinite} proved a limit on the set of distributions $p(a_1, \dots, a_n)$ that can be modeled in a network in the classical setting.  They showed that the covariance matrix of functions $f_1(a_1), \dots, f_n(a_n)$ of output variables associated to such a distribution $p(a_1, \dots, a_n)$ can be decomposed as a summation of positive semidefinite matrices whose terms somehow correspond to sources in the network. Later, the same result is proven by {\AA}berg et al~\cite{aaberg2020semidefinite} in the quantum setting, and discussed in Kraft et al~\cite{kraft2020characterizing}. The main contribution of our paper is that this decomposition holds to exist in a large class of networks in any underlying physical theory satisfying the aforementioned causality principle. More precisely, we prove that in any such physical theory, for local scalar functions $f_1, \dots, f_n$ of outputs of parties, their covariance matrix can be decomposed into a summation of positive semidefinite terms that are associated to sources in the network.  

In the rest of this section, after giving some preliminary definitions, we present the statement of our main result and explain its proof techniques.

\subsection{Main result}\label{subsec:MainResult}

In this paper, we adopt Dirac's ket-bra notation. We let $[d]:=\{1, \dots, d\}$ and denote the computational basis of the $d$-dimensional complex vector space $\mathbb C^d$ by $\{\ket 1, \dots, \ket d\} = \{\ket i:\, i\in [d]\}$. The complex conjugate of $x\in\mathbb C$ is $\bar x$.  Capital letters represent matrices and the $(i, j)$-th entry of matrix $M$ is denoted by $M_{ij}$. For two matrices $M, N$ of the same size, $M\circ N$ is the \emph{Schur or Hadamard product} of $M$ and $N$, that is a matrix with entries 
$$(M\circ N)_{ij} = M_{ij}\cdot N_{ij}.$$
$M\succeq 0$ indicates that $M$ is positive semidefinite.
We represent the Kronecker delta function by $\delta_{k, \ell}$.

\paragraph{Networks.}
A network $\cN$ is a bipartite graph whose vertex set consists of sources and parties. The sources are denoted by $S_\alpha$ for $\alpha=1, \dots, m$, and parties are denoted by $A_i$ for $i=1, \dots, n$. $S_\alpha\to A_i$ (and sometimes $\alpha\to i$) indicates that the source $S_\alpha$ is connected to the party $A_i$. 
We assume throughout this paper that the networks do not have isolated vertices meaning that each source is adjacent to at least one party, and each party is adjacent to at least one source. 
Moreover, we assume  that there are no two sources in the network whose adjacent sets of parties are comparable under inclusion since otherwise one of them is redundant and can be merged with the other one.

\paragraph{Output distribution.} In the network $\cN$, each source $S_\alpha$ sends some \emph{signals} to its adjacent parties, and each party $A_i$ after receiving signals from its adjacent sources, performs a \emph{measurement} and outputs the result $a_i$. The joint distribution of these outputs is denoted by $p(a_1, \dots, a_n)$.

\paragraph{Underlying physical theory.} 
In this paper, we \emph{do not make any assumption about the nature of signals and measurements}.
Hence, we consider any correlation which can be obtained in any Generalized Probabilistic Theory (GPT), i.e., from any hypothetical model for correlations in networks beyond classical and quantum physics.

The causal structure of the Bell scenario can be \emph{derived} from the no-signaling principle alone. This approach fails in general networks. Hence, we need to \emph{accept} the network causal structure. Here, we assume the \emph{causality principle} from which two \emph{laws} are followed.
By the first causality law, the marginal distribution of the outputs of a group of parties depends only on the sources adjacent to them, and not on the structure of the network beyond those sources.
By the second causality law, two groups of parties who do not share a common source are uncorrelated, and their marginal distribution factorizes. We think of these two causality laws as the only generic restrictions on GPTs in networks.\footnote{Our first causality law follows from the the no-signaling principle, and our second causality law is called the \emph{independence principle} in~\cite{gisin2020constraints, bancal2021}. Thus, as an alternative solution we can also take the No-Signaling and the Independence principles (called NSI~\cite{gisin2020constraints, bancal2021}) instead of the causality principle. See~\cite{pironioToBePublished} for a discussion on these different approaches.}
See Figure~\ref{fig:nosignaling}.

We also assume that we can make \emph{independent and identical} copies of the sources and parties of a network $\cN$. That is, we assume that we can repeat the process (experiment) under which the sources produce their signals and the parties perform their measurements.
For instance, for a source $S_\alpha$ of $\cN$ that emits some signal, we can make an independent copy $S_{\alpha'}$ of it that emits the same signal. We emphasize that we do not clone this signal; we only repeat the process under which it is produced.
We give more details on these assumptions and their consequences in Subsection~\ref{subsec:inflation}.

\paragraph{Covariance matrix.}
Suppose that each party $A_i$ applies a real or complex function $f_i(a_i)$ on her output. Associated to the distribution $p(a_1, \dots, a_n)$ and these functions
we consider the \emph{covariance matrix} $\cC=\cC(f_1, \dots, f_n)$. This is an $n\times n$ matrix whose $(i, j)$-th entry equals 
$$\cC_{ij}=\Cov(f_i, f_j) = \E[\bar f_if_j] -\E[\bar f_i]\cdot \E[f_j] = \E\big[  (\bar f_i-\E[\bar f_i])\cdot (f_j-\E[f_j])  \big],$$
where the expectations are with respect to the output distribution $p(a_1, \dots, a_n)$. It is well-known and can easily be verified that the covariance matrix is always positive semidefinite. 

\medskip
In this paper, we are interested in certain decompositions of covariance matrices. 

\begin{definition}[$\cN$-compatible matrix decomposition] \label{def:matrix-decomp}
Let $\cN$ be a network with sources $S_\alpha$, $\alpha=1, \dots, m$, and parties $A_i$, $i=1, \dots, n$. For any source $S_\alpha$ let $\cL_\alpha$ be the set of $n\times n$ matrices $M_\alpha$ such that 
\begin{itemize}
\item[--] $M_\alpha$ is positive semidefinite,
\item[--] The $(i, j)$-th entry of $M_\alpha$ is non-zero only if $\alpha\to i$ and $\alpha\to j$.
\end{itemize}
Then we say that a positive semidefinite $n\times n$ matrix $M$ admits an \emph{$\cN$-compatible matrix decomposition} if there are $M_\alpha\in\cL_\alpha$, $\alpha=1, \dots, m$, such that $M= \sum_\alpha M_\alpha$.
\end{definition}

For an example of a network-compatible matrix decomposition, consider the network $\cN$ of Figure~\ref{fig:nosignaling}. As mentioned in the caption of this figure, by the independence principle $p(a_1, a_3)=p(a_1) p(a_3)$. This means that $\cC_{13}=\bar\cC_{31}=\Cov(f_1, f_3)=0$.  Then an $\cN$-compatible matrix decomposition for the covariance matrix $\cC$ takes the form
\begin{align*}
\cC=\begin{pmatrix}
\Var(f_1) & \Cov(f_1, f_2) & 0\\
\Cov(f_2, f_1) & \Var(f_2) & \Cov(f_2, f_3)\\
0 & \Cov(f_3, f_2) & \Var(f_3)
\end{pmatrix} =  \begin{pmatrix}
\ast & \ast & 0\\
\ast & \ast & 0\\
0 & 0 & 0
\end{pmatrix} +  \begin{pmatrix}
0 & 0 & 0\\
0 & \ast & \ast \\
0 & \ast & \ast
\end{pmatrix},
\end{align*}
where the two matrices on the right hand side are positive semidefinite and belong to $\cL_\alpha$ and $\cL_\beta$, respectively.

In this paper, we prove that the covariance matrix $\cC(f_1, \dots, f_n)$ defined above admits an $\cN$-compatible matrix decomposition with the minimal aforementioned assumptions on the underlying physical theory.
To prove our result we need to restrict to a subclass of networks first introduced in~\cite{renou2020network}.

\begin{definition}[No double common-source network] \label{def:NDCSnetwork}
A network $\cN$ is called a No Double Common-Source (NDCS) network if for any two parties $A_i\neq A_j$, there are \emph{at most one} source $S_\alpha$ adjacent to both of them.
\end{definition}

Observe that a network in which all sources are bipartite (i.e., each source is adjacent to exactly two parties) are automatically NDCS.\footnote{Recall that we assume there is no redundant source in the network.}
We can now formally state the main result of this paper.

\begin{theorem}[Main result]\label{thm:main-result}
Let $\cN$ be an NDCS network with sources $S_\alpha$, $\alpha=1, \dots, m$, and parties $A_i$, $i=1, \dots, n$. Consider signals emitted by sources and measurements performed by the parties in an arbitrary underlying physical theory that satisfies the no-signaling and independence principles. Suppose that this results in an output joint distribution $p(a_1,\dots , a_n)$. 
Let $f_i(a_i)$ be a function of the output of party $A_i$. Then the covariance matrix $\cC(f_1, \dots, f_n)$ admits an $\cN$-compatible matrix decomposition. 
\end{theorem}

In~\cite{kela2019semidefinite} and~\cite{aaberg2020semidefinite} the above theorem is proven for the classical and quantum theories without the NDCS  assumption and for \emph{vector-valued} functions beyond the scalar ones.  Nevertheless, Theorem~\ref{thm:main-result} states the validity of covariance decomposition in the \emph{box world}~\cite{Skrzypczyk2009couplers,Short2010couplers} as well as any GPT~\cite{Barrett2007GPT,Janotta2013norestriction}.

Remark that this theorem is \emph{tight}, in the sense that given an NDCS network $\cN$ and a positive semidefinite matrix $M$ that admits an $\cN$-compatible matrix decomposition, there exists an output distribution of the network whose covariance matrix is $M$. Indeed, letting $M=\sum_\alpha M_\alpha$ be an $\cN$-compatible matrix decomposition of $M$, assume that source $\alpha$ sends a \emph{multivariate Gaussian distribution} with covariance matrix $M_\alpha$ to its adjacent parties (each coordinate to its associated party). Then, if each party outputs the summation of all received (one-dimensional) Gaussian signals, the joint output distribution would be a Gaussian distribution whose covariance matrix is $M$.

\paragraph{Proof ideas.}
Our main conceptual idea in proving Theorem~\ref{thm:main-result} is the so called \emph{inflation technique}~\cite{wolfe2019inflation} (see Subsection~\ref{subsec:inflation}). Using \emph{non-fanout} inflations, we consider certain expansions of the network $\cN$ by adding copies of the sources and parties and connecting them based on similar rules as in $\cN$. Then we consider the covariance matrix associated to the inflated network. As a covariance matrix, it is again positive semidefinite. This fact imposes some conditions on the original covariance matrix. Then we show that these conditions imply that the covariance matrix  admits an $\cN$-compatible matrix decomposition.

From a more technical point of view, we first note that, fixing a network $\cN$, the space of covariance matrices associated to $\cN$ forms a convex cone (see Definition~\ref{def:ConesAndDuals}).
On the other hand, the space of matrices that admit an $\cN$-compatible matrix decomposition is also a convex cone.
To prove Theorem~\ref{thm:main-result} we need to show that these two convex cones coincide.
To this end, based on the theory of \emph{dual cones} (see Subsection~\ref{subsecc:cones}), it suffices to show that the associated dual cones are equal.
Next, to prove equality in the dual picture, we use the inflation technique as discussed above.
Putting these together, the theorem is proven when $\cN$ contains only bipartite sources (see Theorem~\ref{thm:main-result_bipartite}).
For general NDCS networks we also need to use the theory of \emph{embezzlement}, first introduced in~\cite{van2003universal} in the context of entanglement theory (see Subsection~\ref{subsec:embezzlement}).

\section{Main tools}\label{sec:MainTools}

In this section we explain the main three tools we use in the proof of Theorem~\ref{thm:main-result}, namely, non-fanout inflations, the theory of dual cones and the theory of embezzlement.

\subsection{Non-fanout inflation}\label{subsec:inflation}

\begin{figure}
\begin{center}
\includegraphics[scale=0.3]{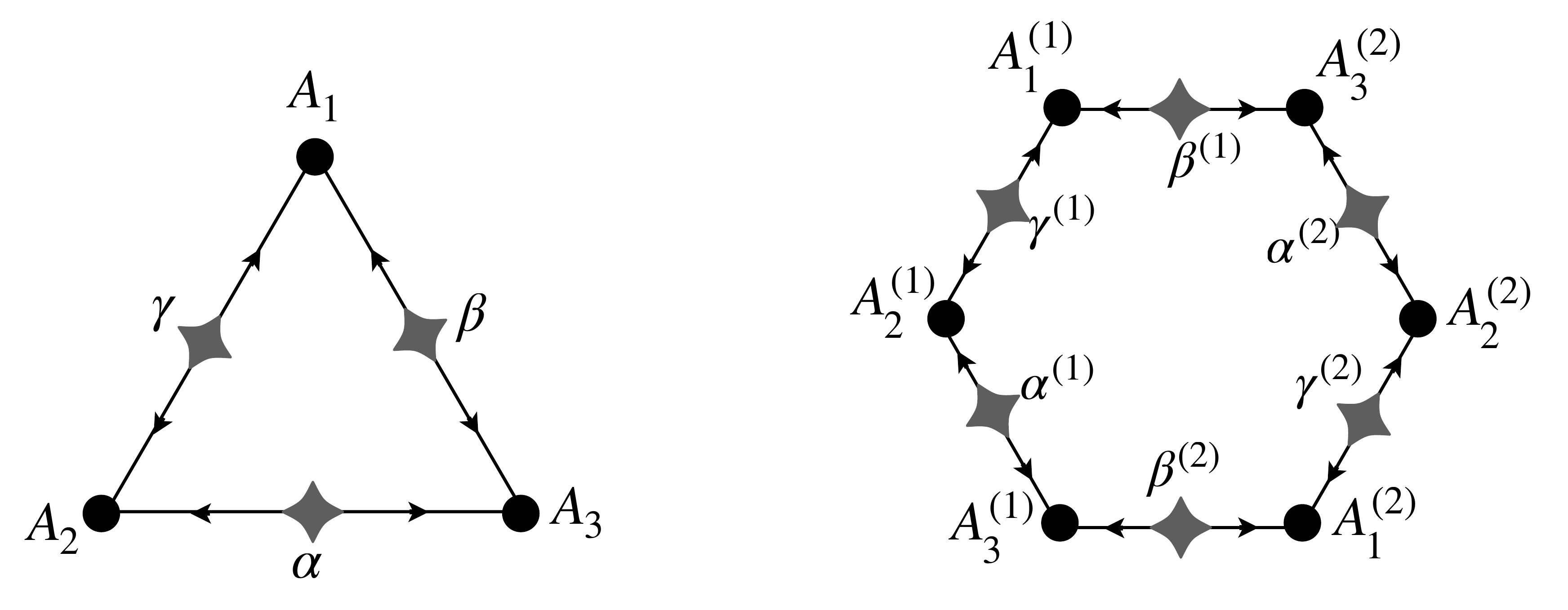}
\caption{\footnotesize The network on the right is a non-fanout inflation of the triangle network on the left. Here, parties $A_i^{(1)}, A_i^{(2)}$ are two copies of $A_i$, for $i=1, 2, 3$. Moreover, sources $\alpha^{(1)}, \alpha^{(2)}$ are two independent copies of $\alpha$, and similarly for $\beta, \gamma$. We note that, e.g., the party $A_1^{(1)}$ is adjacent to 
copies $\beta^{(1)}, \gamma^{(1)}$ of $\beta, \gamma$, respectively, that are adjacent to $A_1$ in the original network. Thus, $A_1^{(1)}$ receives signals of the same types as the ones received by $A_1$ in the triangle, and can behave exactly the same as $A_1$; she can perform the same measurement on the received signals. 
According to the setup of the proof of Lemma~\ref{lem:CovInflationCondition_BipartiteSources}, this inflation corresponds to the sign function $\epsilon(\alpha) =\epsilon(\gamma)=+1$ and $\epsilon(\beta)=-1$.
}
\label{fig:inflation}
\end{center}
\end{figure}

The inflation technique~\cite{wolfe2019inflation} is a method to obtain constraints over the output correlations produced in a network $\cN$.
Here, we give a formal definition of a restricted class of inflations called \emph{non-fanout inflations}. 

Consider a network $\cN$ with sources $S_\alpha$, $\alpha=1, \dots, m$, and parties $A_i$, $i=1, \dots, n$ outputting $a_1, \dots, a_n$ with joint distribution $p(a_1, \dots, a_n)$.
In the inflated network of order $d\geq 2$, which we denote by $\widetilde \cN$, we consider $d$ copies of each source $S_\alpha$ which we denote by $S_\alpha^{(1)}, \dots, S_\alpha^{(d)}$, and $d$ copies of each party $A_i$ which we denote by $A_i^{(1)}, \dots, A_i^{(d)}$.
Let $S_\alpha$ and $A_i$ be an adjacent source and party in $\cN$.
We assume that their associated copies are connected in $\widetilde \cN$ as follows.
Fix some permutation $\pi_i^\alpha$ on $\{1, \dots, d\}$.
Then we assume that $S_\alpha^{(k)}$ is adjacent to $A_i^{(\pi_i^\alpha)^{-1}(k)}$ for any $1\leq k\leq d$. Observe that in $\widetilde \cN$, any party $A_i^{(k)}$ is adjacent to the same number of sources and  of the same types as in $\cN$.
Similarly, any source $S_\alpha^{(k)}$ in $\widetilde \cN$ is adjacent to the same number of parties and of the same types as in $\cN$. Therefore, sources and parties in $\widetilde \cN$ can behave similarly to the sources and parties in the original network by producing the same signals and performing the same measurements. The output of $A_i^{(k)}$ is denoted by $a_i^{(k)}$, and the joint distribution of the outputs is denoted by 
\begin{align}\label{eq:inflated-dist}
p\Big(a_1^{(1)}, \dots, a_1^{(d)}, a_2^{(1)}, \dots, a_2^{(d)}, \dots\dots, a_n^{(1)}, \dots, a_n^{(d)}\Big).
\end{align}
For an example of a non-fanout inflation see Figure~\ref{fig:inflation}.

In this paper, we assume that the underlying physical theory satisfies causality in the network, which is defined through the following two laws:

\paragraph{First causality law.} The marginal distribution of the outputs of a group of parties depends only on the structure of the network near those parties. 
For instance, suppose that $A_i^{(k)}, A_j^{(k')}$ share a source $S_\alpha^{(\ell)}$ in $\widetilde\cN$ and that $S_{\alpha}^{(\ell)}$ is the unique such source.
This means that $S_\alpha \to A_i, A_j$ in $\cN$.
Then the first principle states that the marginal distribution of outputs of $A_i^{(k)}, A_j^{(k')}$ in $\widetilde \cN$ coincides with that of $A_i, A_j$ in $\cN$, i.e., $p\big(a_i^{(k)}, a_j^{(k')}\big) = p(a_i, a_j)$. The point is that the correlation between parties $A_i^{(k)}, A_j^{(k')}$ in $\widetilde \cN$ are created in the very same way as the one between parties $A_i, A_j$ in $\cN$.

As a consequence of this principle, if $A_i, A_j$ compute functions $f_i, f_j$ on their outputs, and similarly $A_i^{(k)}, A_j^{(k')}$ compute identical functions  $f_i^{(k)}, f_j^{(k')}$, respectively, then we have $\Cov(f_i^{(k)}, f_j^{(k')})= \Cov(f_i, f_j)$. 

\paragraph{Second causality law.}
If two groups of parties do not share any common source, then their marginal distribution factorizes.  
For instance, if $A_i^{(k)}, A_j^{(k')}$ do not share any source in $\widetilde\cN$ (which may be the case even if $A_i, A_j$ share a source in $\cN$), then the marginal distribution $p\big(a_i^{(k)}, a_j^{(k')}\big)$ factorizes as $p\big(a_i^{(k)}, a_j^{(k')}\big) = p\big(a_i^{(k)}\big)p\big(a_j^{(k')}\big) = p(a_i)p(a_j)$, where the second equality follows from the first principle. 
This, in particular, means that if $f_i^{(k)}$ and $f_j^{(k')}$ are functions of the outputs of $A_i^{(k)}, A_j^{(k')}$, then we have $\Cov(f_i^{(k)}, f_k^{(k')})=0$. 

\medskip
Importantly, these two laws and inflation are more than a method to find limitations on correlations. They provide the \emph{definition} of the ultimate limits of correlations in networks, hence should be satisfied by any GPT in network (see \cite{coiteuxToBePublished,pironioToBePublished}). 
We emphasize that we only consider non-fanout inflations meaning that we do not clone signals. In the classical theory signals can be cloned, so, e.g.,  a bipartite source in $\cN$ may become tripartite in $\widetilde \cN$. Indeed, the inflation technique with unbounded fanout exactly characterizes distributions in classical networks~\cite{navascues2020inflation}. Here, as we consider arbitrary theories, we restrict only to non-fanout inflations. Such inflations has also been considered in~\cite{gisin2020constraints}.

\subsection{Cones and duality}\label{subsecc:cones}

One of the main tools that we use to prove our main theorem is the theory of dual cones~\cite{dattorro2010convex}. Here, restricting to cones of real or complex \emph{self-adjoint (Hermitian)} matrices, we review the basic definitions and properties. 
Let us first introduce the \emph{Hilbert-Schmidt inner product} on the space of matrices given by
$\langle X, Y\rangle := \tr(X^{\dagger} Y),$
where $X^{\dagger}$ is the adjoint of $X$.
Cones and dual cones are defined as follows:
\begin{definition}\label{def:ConesAndDuals}
A subset $\cK$ of self-adjoint $n\times n$ matrices  is called a \emph{convex cone} if 
\begin{itemize}
\item For $X\in \cK$ and $r\geq 0$ we have $rX\in \cK$,
\item For any $X_1, X_2\in \cK$ we have $X_1+X_2\in \cK$.
\end{itemize}
Moreover, the \emph{dual} cone of $\cK$ is defined by $\cK^*:= \{Y\,|\, \langle X, Y\rangle\geq 0,~~\forall X\in \cK \}.$
\end{definition}

An important example of a convex cone is the set of positive semidefinite matrices, that is self-dual.

In the proof of Theorem~\ref{thm:main-result} we will use the following basic properties of dual cones, whose proofs are left for Appendix~\ref{app:cones}.

\begin{lemma}\emph{\cite{dattorro2010convex} }\label{lem:cones}
The followings hold for both the real and complex scalar fields:
\begin{itemize}
\item[{\rm(i)}] For any cone $\cK$, we have $\cK^* = (\,\overline{\cK}\,)^*$ where $\overline{\cK}$ is the closure of $\cK$. Moreover, the dual cone $\cK^*$ is convex and closed.
\item[{\rm(ii)}] For any two cones $\cL\subseteq \cK$ we have $\cK^*\subseteq \cL^*$.
\item[{\rm(iii)}] For any closed convex cone $\cK$ we have $(\cK^*)^* = \cK$. 
\item[{\rm(iv)}] For any closed convex cones $\cK_1, \cK_2$, we have that $\cK_1\cap \cK_2$ is a closed convex cone.
\item[{\rm(v)}] For any closed convex cones $\cK_1, \cK_2$ , their sum $\cK_1+ \cK_2 = \big\{X_1+ X_2|\, ~ X_1\in \cK_1, X_2\in \cK_2 \big\}$ is a convex cone (but not necessarily closed).
\item[{\rm(vi)}] For any closed convex cones $\cK_1, \cK_2$ we have $(\cK_1\cap \cK_2)^* = \overline{\cK_1^* + \cK_2^*}$.
\item[{\rm(vii)}] For any closed convex cones $\cK_1, \cK_2$ we have $(\cK_1 + \cK_2)^* = \cK_1^* \cap \cK_2^*$. 

\end{itemize}
\end{lemma}

\subsection{Embezzlement}\label{subsec:embezzlement}

\emph{Embezzlement} is another pivotal idea in our proof of Theorem~\ref{thm:main-result}. Entanglement embezzlement is a compelling notion in quantum information theory first introduced in~\cite{van2003universal}. In the following we present the construction of the universal embezzling family of this paper, but without referring to entanglement.  
For any integer $R$ define
\begin{align}\label{eq:def-mu-R}
\ket{\mu_R} := \frac{1}{\sqrt{\chi_R}} \sum_{r=1}^R \frac{1}{\sqrt r}\ket r,
\end{align}
where $\chi_R = \sum_{r=1}^R r^{-1}$ is the Harmonic number
and a normalization factor. 

\begin{lemma} {\rm \cite{van2003universal}} \label{lem:embezzlement}
For any $\varepsilon>0$, integer $d$ and sufficiently large $R$ the following holds. 
Let $\ket{\phi} =\sum_{j=1}^d c_j \ket j$ be a \emph{normal} $d$-dimensional vector satisfying $c_j\geq 0$ for all $j$. Then there exists a permutation $\pi: [dR]\to[dR]$ such that 
$$\bra{\mu_R} P_\pi  \,(\ket{\phi}\otimes \ket{\mu_R}) \geq 1-\eps.$$
Here we identified $[d]\times[R]$ with $[dR]$, e.g., via $(j, r)\mapsto (j-1)R+r$, and $P_{\pi}$ 
is  the permutation matrix associated with $\pi$ given by $P_{\pi}\ket{j}\otimes \ket{r} = \ket{\pi((j-1)R+r)}$.  Note that by abuse of notation we consider $\ket{\mu_R}$ as a vector in both $\mathbb C^{R}$ and $\mathbb C^{dR}$ (by padding it with zero coordinates).
\end{lemma}

For the sake of completeness, we give the proof of this lemma in Appendix~\ref{app:embezzlement}.

We now generalize the above lemma to the case where the coordinates of $\ket{\phi}$ are not necessarily non-negative. For this generalization we will use another class of states. For an integer $T$, let $\omega=e^{2\pi i/T}$ be a $T$-th root of unity and define
\begin{align}\label{eq:def-vartheta-T}
\ket{\vartheta_T}=\frac{1}{\sqrt T} \sum_{t=1}^T \omega^t \ket t.
\end{align}

\begin{lemma}\label{lem:embezzlement-gen}
For any $\varepsilon>0$, integer $d$ and sufficiently large $R, T$ the following holds.
Let $\ket{\phi} =\sum_{j=1}^d c_j \ket j$ be an arbitrary \emph{normal} $d$-dimensional vector. Then there exist a permutation $\pi: [TdR]\to[TdR]$ such that 
$$(\bra{\vartheta_T}\otimes \bra{\mu_R}) P_\pi  \,(\ket{\vartheta_T}\otimes\ket{\phi}\otimes \ket{\mu_R}) \geq 1-\eps.$$
Here as in the previous lemma we identify $[T]\times[d]\times [R]$ with $[TdR]$ via a fixed map, and by abuse of notation we consider $\ket{\vartheta_T}\otimes \ket{\mu_R}$ as a vector in both $\mathbb C^{T}\otimes \mathbb C^R$ and $\mathbb C^{TdR}$ (by padding it with zero coordinates).

\end{lemma}

\begin{proof}
Let $c_j =  b_je^{2\pi i\theta_j} $ with $0\leq \theta_j<1$ and $b_j>0$. Then for any $j$ there exists an integer $ 0\leq r_j<T$ such that
\begin{align}\label{eq:approx-theta-j}
\big|\theta_j - \frac{r_j}{T}\big|\leq \frac{1}{T}.
\end{align}
Let $Q$ be the permutation matrix acting on $\mathbb C^T$ that shifts the basis vectors by 1:
$$Q\ket t = \ket{(t+1)\mod T}, \qquad 1\leq t\leq T.$$
Then we have $Q\ket{\vartheta_T} = \omega^{-1}\ket{\vartheta_T}$. Next, let
$$\widetilde Q = \sum_{j=1}^d Q^{r_j}\otimes \ket j\bra j.$$
Note that $\widetilde Q$ is a permutation matrix acting on $\mathbb C^{d}\otimes \mathbb C^{T}$. Then  we have
\begin{align*}
\widetilde Q \,\ket{\vartheta_T}\otimes \ket{ \phi} & = \sum_{j=1}^d  c_j Q^{r_j}\ket{\vartheta_T}\otimes \ket j\\
& =  \sum_{j=1}^d\omega^{-r_j} c_j \ket{\vartheta_T}\otimes \ket j\\
& =  \sum_{j=1}^d b_je^{2\pi i(\theta_i - r_j/T)} \ket{\vartheta_T}\otimes \ket j.
\end{align*}
Therefore, letting $\ket{\hat \phi} = \sum_{j=1}^d b_j \ket j$ and using~\eqref{eq:approx-theta-j}, for sufficiently large $T$, the vector  $\widetilde Q \,\ket{\vartheta_T}\otimes \ket{ \phi} $ is arbitrarily close to $\ket{\vartheta_T}\otimes \ket{\hat \phi}$. On the other hand, by Lemma~\ref{lem:embezzlement}, there is a permutation matrix $P$ acting on $\ket{\hat \phi}\otimes \ket{\mu_R}$ such that for sufficiently large $R$, the vector $P\ket{\hat \phi} \otimes \ket{\mu_R}$ is sufficiently close to $\ket{\mu_R}$. Putting these together, we find that for sufficiently large $R, T$, the vector $(I\otimes P)(\tilde Q\otimes I)\ket{\vartheta_T}\otimes \ket{\phi}\otimes \ket{\mu_R}$ is sufficiently close to $\ket{\vartheta_T}\otimes \ket{\mu_R}$. We are done since $(I\otimes P)(\tilde Q\otimes I)$ is a multiplication of permutation matrices, and is a permutation matrix itself.

\end{proof}

\section{Bipartite-source networks}\label{sec:BipartiteSourceNetworks}

In this section, as a warm up and to convey some of our ideas, we prove Theorem~\ref{thm:main-result} for networks all of whose sources are bipartite and the functions $f_i$ are real. Such a network can be visualized as a graph with nodes representing parties and edges representing sources. Thus for such networks, a source $S_{\alpha}$ adjacent to two parties $A_i, A_j$ is denoted by $S_{ij}$ and sometimes by $\alpha=(i, j)$.

In the following we restate Theorem~\ref{thm:main-result} for bipartite-source networks and real functions.

\begin{theorem}[Main result, bipartite-source networks]\label{thm:main-result_bipartite}
Let $\cN$ be a network with parties $A_i$, $i=1, \dots, n$ and with sources $S_\alpha$, $\alpha=1, \dots, m$, that are all bipartite. Consider signals emitted by sources and measurements performed by the parties in an arbitrary underlying physical theory that satisfies the no-signaling and independence principles. Suppose that this results in an output joint  distribution $p(a_1,\dots , a_n)$.  
Let $f_i(a_i)$ be a \emph{real} function of the output of party $A_i$. Then the covariance matrix $\cC(f_1, \dots, f_n)$ admits an $\cN$-compatible matrix decomposition. 
\end{theorem}

Our proof of this theorem is in two steps.
First, based on non-fanout inflations, we show in Lemma~\ref{lem:CovInflationCondition_BipartiteSources} that the Schur product of the covariance matrix $\cC(f_1, \dots, f_n)$ with any \emph{sign matrix} (to be defined) is positive semidefinite. Then, using the theory of dual cones, in Lemma~\ref{lem:AlternativeSBMD_BipartiteSources} we show that any matrix with the aforementioned property admits an $\cN$-compatible matrix decomposition. 

\begin{definition}[Sign matrix] 
Consider a network $\cN$ all of whose sources are bipartite.
A sign function $\epsilon$ on $\cN$ is a function that assigns $\epsilon(\alpha)\in \{\pm 1\}$ to any source $S_\alpha$ of $\cN$. 
Then a \emph{sign matrix}  $\Gamma_\epsilon = (\gamma_{ij})$ associated to $\epsilon$ is an $n\times n$ symmetric matrix such that 
\begin{align}
\gamma_{ij} = \begin{cases}
1 &\qquad i=j,\\
\epsilon(\alpha)& \qquad i\neq j ~\&~ \alpha\to i, j.
\end{cases}\label{eq:def-sign-matrix}
\end{align}
\label{def:sign matrix}
Note that if parties $i\neq j$ do not have a common source, then the entry $\gamma_{ij}$ is unconstrained and can take any value.

\end{definition}

We can now state our first lemma.

\begin{lemma}\label{lem:CovInflationCondition_BipartiteSources}
Let $\cN$ be a network with parties $A_1, \dots, A_n$ and with bipartite sources. Let $f_i(a_i)$ be an arbitrary real function of the output of $A_i$. Then for any sign matrix $\Gamma_\epsilon$, the Schur product $\cC(f_1,\dots,f_n)\circ \Gamma_\epsilon$ is positive semidefinite.
\end{lemma}

Note that although the entries $\gamma_{ij}$ of a sign matrix $\Gamma_\epsilon$ for which $i, j$ do not have a common source are not uniquely determined in terms of $\epsilon$, the Schur product $\cC(f_1,\dots,f_n)\circ \Gamma_\epsilon$ depends only on $\epsilon$ (by the independence principle, the corresponding entries in $\cC(f_1,\dots,f_n)$ vanish, forcing those entries in $\cC(f_1,\dots,f_n)\circ \Gamma_\epsilon$ to be zero).

\begin{proof}
To prove this lemma we use the inflation technique discussed in Subsection~\ref{subsec:inflation}. Consider the following non-fanout inflation of $\cN$. For any party $A_i$ in $\cN$ consider two parties  $A_i^{(1)}, A_i^{(2)}$ in the inflated network $\widetilde \cN$. Then for any source $\alpha=(i, j)$ in $\cN$ consider two sources $\alpha^{(1)}, \alpha^{(2)}$ in $\widetilde \cN$. If $\epsilon(\alpha)=+1$, we connect one of them to $A_i^{(1)}, A_j^{(1)}$ and the other one to $A_i^{(2)}, A_j^{(2)}$. If $\epsilon(\alpha)=-1$, then we connect one of them to $A_i^{(1)}, A_j^{(2)}$ and the other one to  $A_i^{(2)}, A_j^{(1)}$. See Figure~\ref{fig:inflation} for an example of a sign function and its associated inflation. 

Now suppose that party $A_i^{(k)}$ computes function $f_i$ of her output, and let $\widetilde \cC$ be the covariance matrix of these $2n$ functions, that is a $(2n)\times (2n)$ positive semidefinite matrix. Suppose that we index rows and columns of $\widetilde\cC$ by $i_k$, $1\leq i\leq n$ and $k=1, 2$, representing party $A_i^{(k)}$. 
Then as mentioned in Subsection~\ref{subsec:inflation} entries of $\widetilde\cC$ can be computed as follows:
\begin{itemize}
    \item $\widetilde \cC_{i_k, i_\ell} = \delta_{k, \ell}\Var(f_i)$,
    \item $\widetilde \cC_{i_k, j_{\ell}} = 0$ if $A_i\neq A_j$ do not have a common source in $\cN$,
    \item $\widetilde \cC_{i_k, j_\ell} =\delta_{k, \ell} \Cov(f_i, f_j)$ if $\alpha=(i, j)$ is a source in $\cN$ and $\epsilon(\alpha)=+1$,
    \item $\widetilde \cC_{i_k, j_\ell} =\delta_{3-k, \ell} \Cov(f_i, f_j)$ if $\alpha=(i, j)$ is a source in $\cN$ and $\epsilon(\alpha)=-1$.
\end{itemize}
Indeed, the covariance matrix $\widetilde{\cC}$ can be written as
\begin{equation}
\widetilde{\cC}=\sum_{\alpha=(i,j)} \Cov(f_i, f_j) \ketbra{i}{j}\otimes T_{\epsilon(\alpha)},
\end{equation}
where 
$$T_{+1}=\begin{pmatrix} 1 & 0 \\ 0 & 1 \end{pmatrix}, \qquad  \qquad
T_{-1}=\begin{pmatrix} 0 & 1 \\ 1 & 0\end{pmatrix}.$$
Next, introducing the Hadamard matrix $$H=\frac{1}{\sqrt{2}}\begin{pmatrix} 1 & 1 \\ 1 & -1 \end{pmatrix},$$
we find that
\begin{equation}\label{eq:cC-hadamard-0}
(I\otimes H) \cdot \widetilde{\cC} \cdot (I\otimes H^\dagger) =\sum_{\alpha=(i,j)} \Cov(f_i, f_j) \ketbra{i}{j}\otimes \begin{pmatrix}
1 & 0 \\
0 & \epsilon(\alpha)
\end{pmatrix}.
\end{equation}
Recall that $\widetilde\cC$ is positive semidefinite. Then the matrix on the right hand side of~\eqref{eq:cC-hadamard-0}, and any of its principal submatrices, is also positive semidefinite. This means that the submatrix consisting of rows and columns indexed by $i_2$, $1\leq i\leq n$, that is equal to 
\begin{equation*}
\sum_{{\alpha}=(i,j)} \Cov(f_i, f_j)\epsilon(\alpha) \ketbra{i}{j} = \cC\circ \Gamma_\epsilon,
\end{equation*}
is positive semidefinite.
\end{proof}

We now show that any matrix satisfying the condition of Lemma~\ref{lem:CovInflationCondition_BipartiteSources} admits an $\cN$-compatible matrix decomposition. 

\begin{lemma}\label{lem:AlternativeSBMD_BipartiteSources}
Let $\cN$ be a network with parties $A_1, \dots, A_n$ and with bipartite sources. Consider an $n\times n$ positive semidefinite matrix $M$ such that $M_{ij}=0$ if $A_i\neq A_j$ do not share a source. 
Then $M$ admits an $\cN$-compatible matrix decomposition if and only if $\widehat M$ is positive semidefinite where
\begin{align*}
\widehat M_{ij} = \begin{cases}
M_{ii} & \quad i=j,\\
-|M_{ij}| & \quad i\neq j,
\end{cases}
\end{align*}
is the \emph{comparison matrix} of $M$.
\end{lemma}

We note that this lemma is first proven in~\cite{ringbauer2018certification} in the context of coherence theory in the special case where the underlying graph is complete (i.e., every two parties share a source). Moreover, the importance of this lemma in the causal inference problem has been notified in~\cite{kraft2020characterizing}. 
This lemma also has applications in entanglement theory~\cite{singh2020ppt}. We will generalise it in Lemma~\ref{lem:dual-cone-description}. Here we present a proof different from those of~\cite{ringbauer2018certification, singh2020ppt}.  

\begin{proof}
We prove this lemma using the theory of dual cones. We first introduce two sets of positive semidefinite matrices as follows:
\begin{itemize}
\item We let $\cK$ be the set of $n\times n$ positive semidefinite matrices $M$ such that $M_{ij}=0$ if $A_i\neq A_j$ do not have a common source in $\cN$, and $\widehat M \succeq 0$.

\item We let $\cL$ be the set of positive semidefinite $n\times n$ matrices that admit an $\cN$-compatible matrix decomposition. Note that using the notation we introduced in Definition~\ref{def:matrix-decomp} we have $\cL=\sum_\alpha \cL_\alpha$.
\end{itemize}
The statement of the lemma says that $\cK=\cL$. As shown in Appendix~\ref{app:K-L-ccc} it is not hard to verify that both $\cK, \cL$ are closed convex cones. Then using Lemma~\ref{lem:cones}, to establish $\cK=\cL$ it suffices to prove that $\cL\subseteq \cK$ and  $\cL^*\subseteq \cK^*$.

\paragraph{$\cL\subseteq \cK$:}
For any $M\in\cL$ there is a decomposition $M=\sum_\alpha M_\alpha$ with $M_\alpha\in \cL_\alpha$. We note that for any $2\times 2$ positive semidefinite matrix $X$ we have $\widehat X\succeq 0$. This means that $\widehat M_\alpha\succeq 0$ for any $\alpha$. Therefore, $\widehat M=\sum_\alpha \widehat M_\alpha$ is positive semidefinite and $ M\in \cK$.

\paragraph{$\cL^*\subseteq \cK^*$:} We need to show that for any $X\in\cL^*$ and $M\in\cK$ we have $\langle X, M\rangle=\tr(X^\dagger M)\geq 0$. To this end, first remark that by Lemma~\ref{lem:cones}, we have $\cL^*=\bigcap_\alpha \cL_\alpha^*$. On the other hand, for $\alpha=(i, j)$, the cone $\cL_\alpha^*$ is the set of matrices whose $\alpha$-block, i.e., the principal submatrix consisting of rows and columns indexed by $i, j$, is positive semidefinite. Therefore, $X\in \cL^*$ means that for any $\alpha=(i,j)$ 
we have $|X_{ij}|\leq\sqrt{X_{ii}X_{jj}}$. Then we have
\begin{align*}
\ttr{X^\dagger M} &= \sum_i M_{ii} X_{ii} + \sum_{\alpha=(i,j)}\big( M_{ij} \bar X_{ij} +M_{ji} \bar X_{ji}\big)\\
 &\geq \sum_i M_{ii} X_{ii} - 2 \sum_{\alpha=(i, j)} |M_{ij}X_{ij}|\\
 &\geq \sum_i M_{ii} X_{ii} - 2 \sum_{\alpha=(i,j)} |M_{ij}| \sqrt{X_{ii}X_{jj}}\\
 &\geq \bra{v} \,\widehat M \,\ket{v}\\
 &\geq 0,
\end{align*}
where $\ket{v}$ is the vector with coordinates $\sqrt{X_{ii}}$, $1\leq i\leq n$. Here the last inequality holds since by assumption $M\in \cK$ and $\widehat M\succeq 0$.
\end{proof}

We can now give the proof of Theorem~\ref{thm:main-result_bipartite}.

\begin{proof}[Proof of Theorem~\ref{thm:main-result_bipartite}]
By Lemma~\ref{lem:AlternativeSBMD_BipartiteSources} we need to show that $\widehat \cC$ is positive semidefinite where $\cC=\cC(f_1, \dots, f_n)$. On the other hand, since by assumption entries of $\cC$ are real, there is a sign matrix $\Gamma_\epsilon$ such that $\widehat \cC = \cC\circ \Gamma_\epsilon$. Next, by Lemma~\ref{lem:CovInflationCondition_BipartiteSources} we know that $\cC\circ \Gamma_\epsilon$ is positive semidefinite. We are done. 

\end{proof}

We remark that  in Theorem~\ref{thm:main-result_bipartite} for simplicity of presentation and conveying some of our ideas we restrict to real functions $f_i$. In the next section where we prove Theorem~\ref{thm:main-result} in its most general form, we cover complex functions $f_i$ as well. Nevertheless, here we briefly explain the changes we should make in the above proofs in order to include complex functions.

Lemma~\ref{lem:AlternativeSBMD_BipartiteSources} already works for complex matrices, so we only need to generalize Lemma~\ref{lem:CovInflationCondition_BipartiteSources}.
To this end, let $\epsilon$ be a function that associates a norm-1 complex number $\epsilon(\alpha)$ to any source $\alpha$.  Then we define the \emph{generalised sign matrix} $\Gamma_\epsilon = (\gamma_{ij})$ to be an $n\times n$ matrix with entries 
\begin{align}
\gamma_{ij} = \begin{cases}
1 &\qquad i=j,\\
\epsilon(\alpha)& \qquad i< j ~\&~ \alpha\to i, j,\\
\bar{\epsilon}(\alpha)& \qquad i> j ~\&~ \alpha\to i, j.
\end{cases}\label{eq:def-gamma-general}
\end{align}
We note that for any positive semidefinite matrix $M$, there is an appropriate $\epsilon$ such that $\widehat M=M\circ \Gamma_\epsilon$.
Thus for complex functions $f_i$, we need to generalize Lemma~\ref{lem:CovInflationCondition_BipartiteSources} and show that $\cC(f_1, \dots, f_n)\circ \Gamma_\epsilon\succeq 0$ for any function $\epsilon(\alpha)$ with $|\epsilon(\alpha)|=1$.
To prove such a lemma, observe that for sufficiently large $d$, there are $0\leq t_\alpha<d$ such that $|\epsilon(\alpha) - \omega^{t_\alpha}|$ is arbitrarily small, where $\omega=e^{2\pi i/d}$ is the $d$-th root of unity.
Then by a continuity argument it suffices  to show that $\cC(f_1, \dots, f_n)\circ \Gamma_\epsilon\succeq 0$ in the special case where $\epsilon(\alpha) = \omega^{t_\alpha}$.
To prove this, we consider an inflation of $\cN$ with $d$ copies per source and party. Letting $\alpha_1, \dots, \alpha_d$ be copies of $\alpha=(i, j)$, we connect $\alpha_k$ to parties $A_i^{(k)}$ and  $A_j^{(k+t_\alpha)}$, the $k$-th and $(k+t_\alpha)$-th copies of $A_i$ and $A_j$ respectively.
Next, we consider the covariance matrix associated to this inflated network that is an $(nd)\times (nd)$ positive semidefinite matrix.
The $d\times d$ block of this matrix associated to the source $\alpha=(i,j)$ equals $\Cov(f_i, f_j)$ times a permutation matrix that is a shift by $t_\alpha$.
Then, conjugating this $(nd)\times (nd)$ matrix with the block-diagonal matrix all of whose diagonal blocks are Fourier transform, we can simultaneously diagonalize all blocks of the $(nd)\times (nd)$ covariance matrix. Next, putting the second diagonal element of these blocks together we obtain an $n\times n$ matrix that is equal to $\cC(f_1, \dots, f_n)\circ \Gamma_\epsilon$.
Thus, this matrix is positive semidefinite since it is a principal submatrix of a positive semidefinite matrix.    This gives the proof of Lemma~\ref{lem:CovInflationCondition_BipartiteSources} for complex functions.

In the next section, building on the above ideas, we generalize Theorem~\ref{thm:main-result_bipartite} for arbitrary NDCS networks and for arbitrary complex functions.

\section{NDCS networks}\label{sec:NDCSNetworks}

In this section we prove Theorem~\ref{thm:main-result} in its most general form. 
As in the previous section, the proof is in two steps.
We first show in Lemma~\ref{lem:gen-inflation} (which generalises Lemma~\ref{lem:CovInflationCondition_BipartiteSources}) that the Schur product of the covariance matrix $\cC(f_1, \dots, f_n)$ with any \emph{twisted Gram matrix} (to be defined) is positive semidefinite. We prove this using non-fanout inflations. Then in Lemma~\ref{lem:dual-cone-description} (which generalises Lemma~\ref{lem:AlternativeSBMD_BipartiteSources}) we show that any matrix with the aforementioned property admits an $\cN$-compatible matrix decomposition. To prove this result we use the theory of dual cones and the idea of embezzlement discussed in Subsection~\ref{subsec:embezzlement}.

Let us first generalise the notion of sign matrices. Here, the sign function in Definition~\ref{def:sign matrix} is replaced by a collection of permutations and a set of vectors.

\begin{definition}[Twisted Gram matrix]
Let $\cN$ be an arbitrary NDCS network with $n$ parties $A_1, \dots, A_n$. Let $\ket{\psi_1}, \dots, \ket{\psi_n}\in \mathbb C^d$ be arbitrary $d$-dimensional vectors. Also for any source $\alpha$ and parties $A_i$ with $\alpha\to i$, let $\pi_i^\alpha$ be an arbitrary permutation on $\{1, \dots, d\}$. Let $P_{\pi_i^\alpha}$ be the associated permutation matrix acting on $\mathbb C^d$ (i.e., $P_{\pi_i^\alpha}\ket x= \ket{\pi_i^\alpha(x)}$ for any $1\leq x\leq d$). Then a twisted Gram matrix associated to these vectors and permutations is an $n\times n$ matrix $W$ such that
\begin{align}
W_{ij} = \begin{cases}
\bra{\psi_i}\psi_i\rangle &\qquad i=j,\\
\bra{\psi_i} P^{\dagger}_{\pi_i^\alpha} P_{\pi_j^\alpha}\ket{\psi_j}& \qquad i\neq j ~\&~ \alpha\to i, j.
\end{cases}\label{eq:def-W}
\end{align}
Note that, when parties $i\neq j$ do not share any common source, $W_{ij}$ is unconstrained and can take any value. Also, note that $W$ is well-defined since by the NDCS assumption if $A_i, A_j$ have a common source, then it is unique.
\label{def:twisted-Gram-matrix}
\end{definition}

We can now generalise Lemma~\ref{lem:CovInflationCondition_BipartiteSources}.

\begin{lemma}\label{lem:gen-inflation}
Let $\cN$ be an arbitrary NDCS network with $n$ parties $A_1, \dots, A_n$. Let $f_i(a_i)$, $i=1, \dots, n$, be an arbitrary scalar function of the output of $A_i$. Then for any twisted Gram matrix $W$, the Schur product $\cC(f_1, \dots, f_n)\circ W$ is positive semidefinite.
\end{lemma}   

\begin{proof} Let $W$ be a twisted Gram matrix given by~\eqref{eq:def-W}. Note that the entries of $W$ that are not specified by~\eqref{eq:def-W} are not important since by the independence principle the corresponding entries in $\cC(f_1, \dots, f_n)$ vanish, so $\cC(f_1, \dots, f_n)\circ W$ is independent of those entries of $W$. 
We consider a non-fanout inflation of the network as follows. For any party $A_i$ we consider $d$ copies $A_i^{(1)}, \dots, A_i^{(d)}$, and for any source $S_\alpha$ we consider $d$ copies $S_{\alpha}^{(1)}, \dots, S_{\alpha}^{(d)}$. If source $S_\alpha$ is adjacent to party $A_i$ in $\cN$, in the inflated network we connect source $S_{\alpha}^{(k)}$ to the party $A_i^{((\pi_i^\alpha)^{-1}(k))}$. This fully describes the inflated network. Next, we assume that party $A_i^{(k)}$ computes $f_i^{(k)}$ that is identical to $f_i$. Then the covariance matrix 
$$\widetilde\cC = \cC\big(f_1^{(1)}, \dots, f_1^{(d)}, f_2^{(1)}, \dots, f_2^{(d)}, \dots\dots, f_n^{(1)}, \dots, f_n^{(d)}\big),$$
is positive semidefinite. Observe that $\widetilde\cC$ is an $(nd)\times (nd)$ matrix consisting of $d^2$ blocks of size $n\times n$. These blocks labeled by pairs $(i, j)$ with $1\leq i,j\leq n$ and denoted by $\widetilde\cC_{ij}$ are described as follows:
\begin{itemize}
\item[--] The $(i, i)$-block is diagonal with $\Var(f_i)=\Cov(f_i, f_i)$ on the diagonal, i.e., $\widetilde\cC_{ii}= \Var(f_i) I$. 

\item[--] If $i\neq j$ and parties $A_i, A_j$ do not share any source, then the $(i, j)$-block is $\widetilde\cC_{ij}=0$.

\item[--] If $i\neq j$ and parties $A_i, A_j$ share source $S_\alpha$, then the $(i, j)$-block equals 
$\widetilde\cC_{ij}=\Cov(f_i, f_j) P_{\pi_i^{\alpha}}^\dagger P_{\pi_j^\alpha}$. See Figure~\ref{fig:permutation} to verify this. 

\end{itemize}

\begin{figure}
\begin{center}
\includegraphics[scale=0.3]{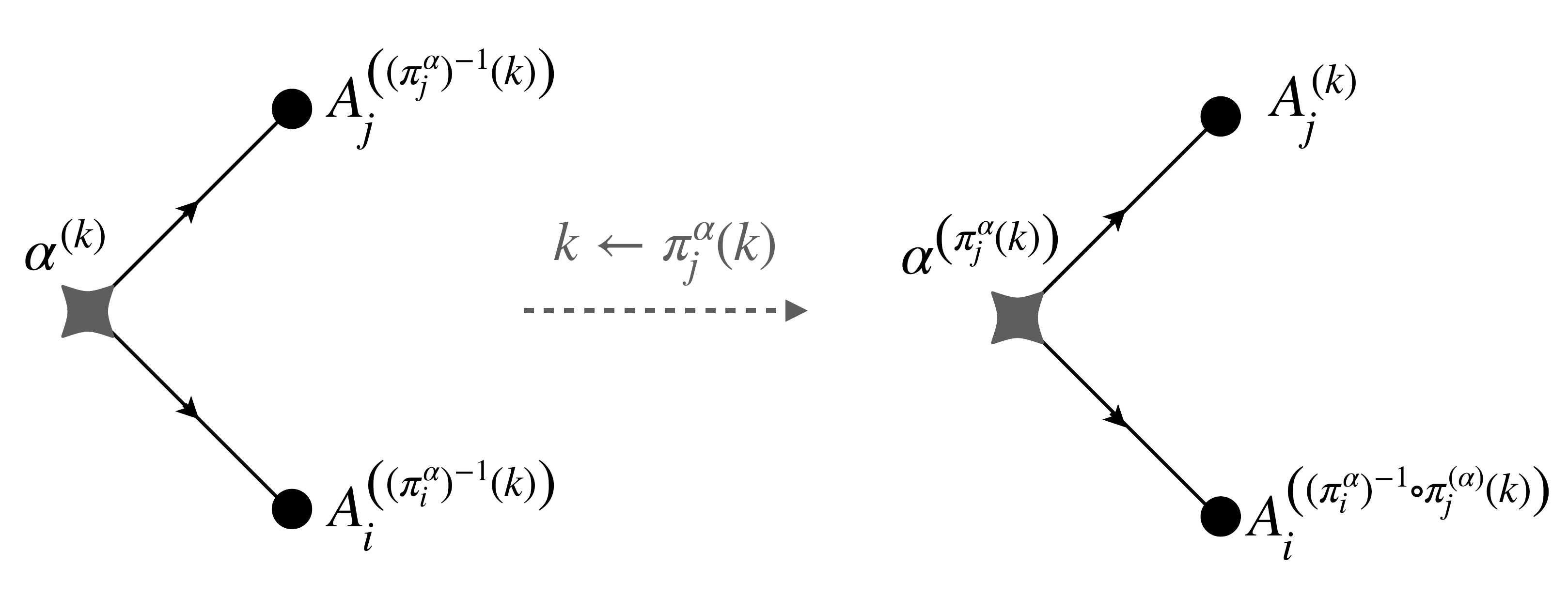}
\caption{\footnotesize If $\alpha\to A_i, A_j$ in $\cN$, then in the inflated network $S_{\alpha}^{(k)}$ is adjacent to $A_i^{((\pi_i^\alpha)^{-1}(k))}$ and $A_j^{((\pi_j^\alpha)^{-1}(k))}$ for any $k$. This means that $S_{\alpha}^{(\pi_j^\alpha(k))}$ is adjacent to $A_j^{(k)}$ and $A_i^{((\pi_i^\alpha)^{-1}\circ \pi_j^\alpha(k))}$.
}
\label{fig:permutation}
\end{center}
\end{figure}  

Let $R_i$ be a $d\times d$ matrix such that 
$$R_iR_i^\dagger = R_i^\dagger R_i=\|\psi_i\|^2I \qquad \quad R_i\ket i=\ket{\psi_i}.$$
Such a matrix can be constructed as follows. If $\ket{\psi}=0$, then let $R_i=0$. Otherwise, start with $\frac{1}{\|\psi_i\|} \ket{\psi_i}$ and extend it to an orthonormal basis; put those basis vectors in the columns of a matrix to get a unitary matrix; finally multiply this unitary by $\|\psi_i\|$ to obtain $R_i$.
Next, let $R$ be the $(nd)\times (nd)$ block diagonal matrix whose $i$-th block on the diagonal is $R_i$.
Then $R^\dagger \widetilde\cC R$ is a positive semidefinite matrix consisting of blocks $R_i^\dagger \widetilde\cC_{ij}R_j$.
Therefore, the principal submatrix of $R^\dagger \widetilde\cC R$ consisting of the $(1, 1)$-entry of all these blocks is also a positive semidefinite matrix. Denoting this $n\times n$ matrix by $M$, the entries of $M$ are computed as follows:
\begin{itemize}
\item[--] $M_{ii}= (R_i^\dagger \,\widetilde\cC_{ii}\,R_i)_{11} = \Var(f_i) \|\psi_i\|^2$. 

\item[--] If $i\neq j$ and parties $A_i, A_j$ do not share any source, then the $(i, j)$-block is $M_{ij}=0$.

\item[--] If $i\neq j$ and parties $A_i, A_j$ share source $S_\alpha$, then 
$M_{ij}=\Cov(f_i, f_j) \bra{\psi_i}P_{\pi_i^{\alpha}}^\dagger P_{\pi_j^\alpha}\ket{\psi_j}$.

\end{itemize}
Therefore, 
$$M= \cC(f_1, \dots, f_n) \circ W,$$
and it is a positive semidefinite matrix. We are done.

\end{proof}

In the following lemma we present a description of the dual cone of positive semidefinite matrices that admit a network-compatible matrix decomposition. This lemma may be of independent interest, notably in coherence and entanglement theory.

\begin{lemma}\label{lem:dual-cone-description}
Let $\cN$ be an arbitrary NDCS network with $n$ parties $A_i$, $i=1, \dots, n$ and sources $S_\alpha$, $\alpha=1, \dots, m$. Let $\cL$ be the cone of positive semidefinite matrices that admit an $\cN$-compatible matrix decomposition. Then 
$$\cL^{*} = \overline{\cW},$$ 
where $\cW$ is the set of twisted Gram matrices of Definition~\ref{def:twisted-Gram-matrix}, and $\overline{\cW}$ is its closure.
\end{lemma}

Note that this lemma is a generalization of Lemma~\ref{lem:AlternativeSBMD_BipartiteSources}, expressed in the dual picture. In the bipartite case, twisted Gram matrices turn to (generalized) sign matrices which have a much simpler description. Then the dual cone $\cW^*$, that is the cone of interest, is easily characterized; it contains all matrices whose comparison matrix is positive semidefinite. This provides the explicit characterization of $\cL$ in Lemma~\ref{lem:AlternativeSBMD_BipartiteSources}.
In the multipartite case, the set $\cW$ takes a much more complicated structure, so characterizing the dual cone $\cW^*$ is not easy. To clarify this issue, note that assuming that all sources are bipartite, the permutations that we need to use are either transpositions (in the real case) or shifts (in the complex case). Such permutations can be simultaneously diagonalized using Fourier transform as we did in the proof of Lemma~\ref{lem:CovInflationCondition_BipartiteSources} and mentioned at the end of Section~\ref{sec:BipartiteSourceNetworks}. 
However, in the multipartite case, we need to use arbitrary permutations that cannot be simultaneously diagonalized. This makes it difficult to obtain a simple direct description of the cone $\cW$.

\begin{proof}
For any $n\times n$ matrix $M$ we let $M^{(\alpha)}$ be the principal submatrix of $M$ consisting of rows and columns indexed by parties $i$ with $\alpha\to i$. We call $M^{(\alpha)}$ the \emph{$\alpha$-block} of $M$. 
Let $\cL_\alpha$ be the cone of positive semidefinite matrices $M$ whose only non-zero entries are in their $\alpha$-block, i.e.,  $M_{ij}\neq 0$ only if $\alpha\to i, j$. We note that as the dual of the cone of positive semidefinite matrices is itself, $\cL_\alpha^*$ is the space of matrices whose $\alpha$-block is positive semidefinite. On the other hand, 
by definition 
$\cL = \sum_\alpha \cL_\alpha$,
so by Lemma~\ref{lem:cones} we have
$$\cL^* = \bigcap_{\alpha} \cL_{\alpha}^*.$$
That is, $\cL^*$ consists of matrices whose $\alpha$-block for any source $S_\alpha$ is positive semidefinite.

\paragraph{$\overline{\cW}\subseteq \cL^*$:}
By definition, any $\alpha$-block of any matrix $W\in \cW$ is a Gram matrix and is positive semidefinite. Then by the above discussion, we have $\cW\subseteq \cL^*$. On the other hand $\cL^*$, as a dual cone, is closed. Therefore, $\overline{\cW}\subseteq \cL^*$. 

\paragraph{$\cL^*\subseteq \overline{\cW}$:}
Let $M\in \cL^*$. In the following we show that $M$ can be approximated by elements of $\cW$ with arbitrary precision, which shows that $M\in \overline{\cW}$. To this end, we note that if $i\neq j$ do not share any source, the $(i,j)$-entries of matrices in $ \cW$ can take any value. So such entries of $M$ can be ignored in approximating $M$ with elements of $\cW$. Indeed, it suffices to consider entries of $M$ that belong to an $\alpha$-block for some source $\alpha$.

Since by the characterization of $\cL^*$, the $\alpha$-block of $M$ is positive semidefinite, there are vectors $\ket{\phi_i^\alpha}$, for any $i$ with $\alpha\to i$, such that 
$$M_{ij} = (M^{(\alpha)})_{ij} = \bra{\phi_i^\alpha}\phi_j^\alpha\rangle. $$
By padding these vectors with zero coordinates if necessary, we assume that $\ket{\phi_i^\alpha}\in \mathbb C^d$ for all $i, \alpha$, for some $d$. 
To prove the lemma we need to show that for any $\eps>0$ there are vectors $\ket{\psi_i}$ and permutations $\pi_i^\alpha$ such that
$$\big|\bra{\psi_i} P_{\pi_i^\alpha}^\dagger P_{\pi_j^\alpha}\ket{\psi_j} - \bra{\phi_i^\alpha}\phi_j^\alpha\rangle     \big|\leq \eps.$$
To construct these vectors we use Lemma~\ref{lem:embezzlement-gen}. 
Let 
$$\ket{\psi_i} = \|\phi_i^\alpha\| \cdot \ket{\vartheta_T}\otimes \ket{\mu_R} = \sqrt{M_{ii}} \ket{\vartheta_T}\otimes \ket{\mu_R}.$$
where $\ket{\mu_R}$ and $\ket{\vartheta_T}$ are defined in~\eqref{eq:def-mu-R} and~\eqref{eq:def-vartheta-T} respectively. 
By Lemma~\ref{lem:embezzlement-gen}  there exists a permutation $\pi_i^\alpha$ such that 
$$\bra{\psi_i} P_{\pi_i^\alpha}\big(\ket{\vartheta_T}\otimes \ket{\phi_i^\alpha}\otimes \ket{\mu_R}\big)\geq M_{ii} \,\kappa_{T,R}, $$
where $\kappa_{T, R}\to 0$ as $T, R\to \infty$.
This means that for sufficiently large $T, R$, the vector $P_{\pi_i^\alpha}\ket{\psi_i}$ is arbitrarily close to $\ket{\vartheta_T}\otimes\ket{\phi_i^\alpha}\otimes \ket{\mu_R}$. Therefore, for sufficiently large $T, R$ we have
$$\eps\geq \big|\bra{\psi_i} P_{\pi_i^\alpha}^\dagger P_{\pi_j^\alpha}\ket{\psi_j} - \langle \vartheta_T|\vartheta_T\rangle\cdot \bra{\phi_i^\alpha}\phi_j^\alpha\rangle  \cdot \langle \mu_R|\mu_R\rangle   \big|= \big|\bra{\psi_i} P_{\pi_i^\alpha}^\dagger P_{\pi_j^\alpha}\ket{\psi_j} - \bra{\phi_i^\alpha}\phi_j^\alpha\rangle     \big|.$$
We are done.

\end{proof}

Putting Lemma~\ref{lem:gen-inflation} and Lemma~\ref{lem:dual-cone-description} together, we can now prove Theorem~\ref{thm:main-result}.

\begin{proof}[Proof of Theorem~\ref{thm:main-result}]
Using the notation of Lemma~\ref{lem:dual-cone-description} we need to show that 
$$\cC=\cC(f_1, \dots, f_n) \in \cL.$$
To this end, using $(\cL^*)^*=\cL$, it suffices to show that for any $W\in \cL^*$ we have 
$$\langle \cC, W\rangle = \tr(\cC W)\geq 0.$$ 
Then using $\cL^*=\overline \cW$ established in Lemma~\ref{lem:dual-cone-description} and by a continuity argument, we may assume that $W\in \cW$ and takes the form~\eqref{eq:def-W}.  Now letting $\ket{J} = \sum_{j=1}^n \ket j$, we have
$$\langle \cC, W\rangle =\bra J \cC\circ W\ket J. $$
On the other hand, by Lemma~\ref{lem:gen-inflation} the matrix $\cC\circ W$ is positive semidefinite. Therefore, $\langle \cC, W\rangle\geq 0$.

\end{proof}

\section{Conclusion}

In this paper we showed that the covariance matrix of any output distribution of NDCS networks can be written as a summation of certain positive semidefinite matrices. This result holds in the local classical theory, the quantum theory, the \emph{box world}~\cite{Skrzypczyk2009couplers,Short2010couplers} and more generally in any GPT~\cite{Barrett2007GPT,Janotta2013norestriction, Weilenmann2020SelfTestingQMech} compatible with the network structure. 
To our knowledge, our result provides the first universal limit on correlations in \emph{generic networks}, beyond constraints derived in specific ones (see \cite{gisin2020constraints, bancal2021}).
As also mentioned in~\cite{kela2019semidefinite, aaberg2020semidefinite} this covariance decomposition condition can be stated as a semidefinite program and can  be verified efficiently. 

Our proof technique is valid only for the subclass of NDCS networks, while the covariance decomposition is known to hold for all networks in the case of local and quantum models. Hence, the necessity of this NDCS assumption in the case of generic GPTs is an open question; does the network-compatible matrix decomposition holds in all GPTs for networks that do not satisfy the NDCS assumption? Another open problem is to generalize our results for vector-valued functions beyond scalar ones.

The main conceptual tool in our proof is the inflation technique.
It is proven in~\cite{navascues2020inflation} that inflations with unbounded fanout characterize correlations in the local classical model. 
On the other hand, non-fanout inflation allows for a characterization of the ultimate limits on correlations in networks~\cite{gisin2020constraints,coiteuxToBePublished,pironioToBePublished}.
Our results show that even non-fanout inflations induce quite powerful limits on correlations in networks. To our knowledge, our work contains the first proof exploiting inflated networks of arbitrary sizes.

Note that the covariance matrix of a correlation depends only on its bipartite marginals. Yet, even the simple hexagon inflation of the triangle network (see Figure~\ref{fig:inflation}) allows for proving restrictions on multipartite marginals~\cite{gisin2020constraints}.
Thus, it would be interesting to find generic limits imposed on multipartite correlation functions by considering well-chosen families of non-fanout inflations, possibly of unbounded size.
For instance, it would be interesting to see if the Finner's inequality~\cite{finner1992}, as a constraint on multipartite correlations,  holds for arbitrary GPTs. 
This inequality is proven in~\cite{renou2019limits} for the quantum as well as the box world when the underlying network is the triangle.

At last, let us mention that our results may be of interest beyond causal inference in networks, particularly in coherence and entanglement theory. In Lemma~\ref{lem:AlternativeSBMD_BipartiteSources} and Lemma~\ref{lem:dual-cone-description} we characterized the space of matrices that admit a network-compatible matrix decomposition. Lemma~\ref{lem:AlternativeSBMD_BipartiteSources} has applications in coherence and entanglement theory. Thus, Lemma~\ref{lem:dual-cone-description} as a generalization of this lemma may find applications in these theories or elsewhere. Moreover, the proof of this theorem shows that the idea of embezzlement may find other applications in matrix analysis beyond the entanglement theory.

\section*{Acknowledgements.} The authors are thankful to Amin Gohari for bringing the example of Gaussian sources to their attention to show the tightness of Theorem~\ref{thm:main-result}; and to Elie Wolfe, Stefano Pironio and Nicolas Gisin for clarifying the different ways to define the ultimate limits of correlations in networks.
M.-O.R.~is supported by the Swiss National Fund Early Mobility Grant P2GEP2\_191444 and acknowledges the Government of Spain (FIS2020-TRANQI and Severo Ochoa CEX2019-000910-S), Fundació Cellex, Fundació Mir-Puig, Generalitat de Catalunya (CERCA, AGAUR SGR 1381) and the ERC AdG CERQUTE.


{\small
\begin{thebibliography}{10}

\bibitem{aaberg2020semidefinite}
J.~{\AA}berg, R.~Nery, C.~Duarte, and R.~Chaves.
\newblock Semidefinite tests for quantum network topologies.
\newblock {\em Physical Review Letters}, 125(11):110505, 2020.

\bibitem{Allen2017QcommonCauses}
J.-M.~A.~Allen, J.~Barrett, D.~C.~Horsman, C.~M.~Lee, and R.~W.~Spekkens.
\newblock Quantum common causes and quantum causal models.
\newblock {\em Phys. Rev. X}, 7:031021, Jul 2017.

\bibitem{bancal2021}
J.-D. Bancal and N.~Gisin.
\newblock Non-local boxes for networks.
\newblock {\em arXiv preprint arXiv:2102.03597}, 2021.

\bibitem{Barrett2007GPT}
J.~Barrett.
\newblock Information processing in generalized probabilistic theories.
\newblock {\em Phys. Rev. A}, 75:032304, Mar 2007.

\bibitem{bell1964einstein}
J.~S. Bell.
\newblock On the Einstein Podolsky Rosen paradox.
\newblock {\em Physics Physique Fizika}, 1(3):195, 1964.

\bibitem{Branciard2010}
C.~Branciard, N.~Gisin, and S.~Pironio.
\newblock {Characterizing the Nonlocal Correlations Created via Entanglement
  Swapping}.
\newblock {\em Phys. Rev. Lett.}, 104:170401, Apr 2010.

\bibitem{coiteuxToBePublished}
X.~Coiteux-Roy, E.~Wolfe, and M.-O. Renou.
\newblock In preparation.
\newblock 2021.

\bibitem{dattorro2010convex}
J.~Dattorro.
\newblock {\em Convex optimization \& Euclidean distance geometry}.
\newblock Meboo Publishing, 2010.

\bibitem{finner1992}
H.~Finner.
\newblock A Generalization of Holder's Inequality and Some Probability Inequalities.
\newblock {\em The Annals of probability}, 20:1893--1901, 1992.

\bibitem{Fritz2012BeyondBell}
T.~Fritz.
\newblock Beyond Bell's theorem: correlation scenarios.
\newblock {\em New Journal of Physics}, 14(10):103001, oct 2012.

\bibitem{FritzBeyondBell2}
T.~Fritz.
\newblock Beyond Bell's theorem ii: Scenarios with arbitrary causal structure.
\newblock {\em Communications in Mathematical Physics}, 341:391–434, 2016.

\bibitem{gisin2020constraints}
N.~Gisin, J.-D. Bancal, Y.~Cai, P.~Remy, A.~Tavakoli, E.~Z. Cruzeiro,
  S.~Popescu, and N.~Brunner.
\newblock Constraints on nonlocality in networks from no-signaling and
  independence.
\newblock {\em Nature Communications}, 11(1):1--6, 2020.

\bibitem{HensonLimitCorrelationGeneralisedBayNet}
R.~L. J.~Henson and M.~F. Pusey.
\newblock Theory-independent limits on correlations from generalized bayesian
  networks.
\newblock {\em New Journal of Physics}, 16:113043, 2014.

\bibitem{Janotta2013norestriction}
P.~Janotta and R.~Lal.
\newblock Generalized probabilistic theories without the no-restriction
  hypothesis.
\newblock {\em Phys. Rev. A}, 87:052131, May 2013.

\bibitem{kela2019semidefinite}
A.~Kela, K.~Von~Prillwitz, J.~{\AA}berg, R.~Chaves, and D.~Gross.
\newblock Semidefinite tests for latent causal structures.
\newblock {\em IEEE Transactions on Information Theory}, 66(1):339--349, 2019.

\bibitem{kraft2020characterizing}
T.~Kraft, C.~Spee, X.-D. Yu, and O.~G{\"u}hne.
\newblock Characterizing quantum networks: Insights from coherence theory.
\newblock {\em arXiv preprint arXiv:2006.06693}, 2020.

\bibitem{navascues2020inflation}
M.~Navascu{\'e}s and E.~Wolfe.
\newblock The inflation technique completely solves the causal compatibility
  problem.
\newblock {\em Journal of Causal Inference}, 8(1):70--91, 2020.

\bibitem{Pearl2009Causality}
J.~Pearl.
\newblock {\em Causality}.
\newblock Cambridge University Press, 2009.

\bibitem{pironioToBePublished}
S.~Pironio.
\newblock In preparation.
\newblock 2021.

\bibitem{Popescu1994}
S.~Popescu and D.~Rohrlich.
\newblock {Quantum Nonlocality as an Axiom}.
\newblock {\em Found. Phys.}, 24(3):379--385, 1994.

\bibitem{popescu1998causality}
S.~Popescu and D.~Rohrlich.
\newblock Causality and nonlocality as axioms for quantum mechanics.
\newblock In {\em Causality and Locality in Modern Physics}, pages 383--389.
  Springer, 1998.

\bibitem{renou2020network}
M.-O. Renou and S.~Beigi.
\newblock Network nonlocality via rigidity of token-counting and
  color-matching.
\newblock {\em arXiv preprint arXiv:2011.02769}, 2020.

\bibitem{renou2019limits}
M.-O. Renou, Y.~Wang, S.~Boreiri, S.~Beigi, N.~Gisin, and N.~Brunner.
\newblock Limits on correlations in networks for quantum and no-signaling
  resources.
\newblock {\em Physical Review Letters}, 123(7):070403, 2019.

\bibitem{ringbauer2018certification}
M.~Ringbauer, T.~R. Bromley, M.~Cianciaruso, L.~Lami, W.~S. Lau, G.~Adesso,
  A.~G. White, A.~Fedrizzi, and M.~Piani.
\newblock Certification and quantification of multilevel quantum coherence.
\newblock {\em Physical Review X}, 8(4):041007, 2018.

\bibitem{Scarani2006PRbox}
V.~Scarani.
\newblock {Feats, Features and Failures of the {PR}-box}.
\newblock In {\em {AIP} Conference Proceedings}. {AIP}, 2006.

\bibitem{Short2010couplers}
A.~J. Short and J.~Barrett.
\newblock Strong nonlocality: a trade-off between states and measurements.
\newblock {\em New J. Phys.}, 12(3):033034, 2010.

\bibitem{singh2020ppt}
S.~Singh and I.~Nechita.
\newblock The ppt$^2$ conjecture holds for all choi-type maps.
\newblock {\em arXiv preprint arXiv:2011.03809}, 2020.

\bibitem{Skrzypczyk2009couplers}
P.~Skrzypczyk and N.~Brunner.
\newblock Couplers for non-locality swapping.
\newblock {\em New J. Phys.}, 11(7):073014, 2009.

\bibitem{Spirtes2000CausationPrediction}
P.~Spirtes, C.~Glymour, R.~Scheines, D.~Heckerman, C.~Meek, G.~Cooper, and
  T.~Richardson.
\newblock {\em Causation, prediction, and search}.
\newblock MIT Press, 2000.

\bibitem{van2003universal}
W.~van Dam and P.~Hayden.
\newblock Universal entanglement transformations without communication.
\newblock {\em Physical Review A}, 67(6):060302, 2003.

\bibitem{Weilenmann2020SelfTestingQMech}
M.~Weilenmann and R.~Colbeck.
\newblock Self-testing of physical theories, or, is quantum theory optimal with
  respect to some information-processing task?
\newblock {\em Phys. Rev. Lett.}, 125:060406, Aug 2020.

\bibitem{wolfe2019inflation}
E.~Wolfe, R.~W. Spekkens, and T.~Fritz.
\newblock The inflation technique for causal inference with latent variables.
\newblock {\em Journal of Causal Inference}, 7(2), 2019.

\end{thebibliography}
}


\appendix

\section{Proof of Lemma~\ref{lem:cones}}\label{app:cones}
Items (i), (ii), (iv) and (v) follow from the definition of dual cones. 

To prove (iii), note that by definition $\cK\subseteq (\cK^*)^*$. To prove the inclusion in the other direction, let $Z\notin \cK$. Since $\cK$ is convex, by the Hahn-Banach separation theorem there exists $Y$ and $c\in \mathbb R$ such that $\langle Y, X\rangle\geq c$ for all $X\in \cK$ and $\langle Y, Z\rangle<c$. We note that by the definition of cones, $X=0$ belongs to $\cK$. Therefore, $c\leq 0=\langle Y, 0\rangle$. On the other hand, suppose that $X\in \cK$ is such that $x=\langle Y, X\rangle <0$. Then letting $r=(c-1)/x>0$, we have $rX\in \cK$. Therefore,
$c\leq \langle Y, rX\rangle = r\langle Y, X\rangle = c-1$ which is a contradiction. This means that for any $X\in \cK$ we have $\langle Y, X\rangle \geq 0$. As a result, $Y\in \cK^*$ while we have $\langle Y, Z\rangle<c\leq 0$. This means that $Z$ does not belong to $(\cK^*)^*$. Therefore, the complement of $\cK$ is a subset of the complement of $(\cK^*)^*$ which means that $(\cK^*)^*\subseteq \cK$.

To prove (vi) note that by (ii) we have $\cK_i^*\subseteq (\cK_1\cap \cK_2)^*$ for $i=1,2$. Then since $(\cK_1\cap \cK_2)^*$ is a closed convex cone and closed under summation, we have $\overline{\cK_1^*+\cK_2^*}\subseteq (\cK_1\cap \cK_2)^*$. To prove the other direction, suppose that $Z\notin \overline{\cK_1^*+\cK_2^*}$. Then by the Hahn-Banach separation theorem and following similar steps as in the proof of (iii) we find that there exists $Y$ such that $\langle Y, X\rangle\geq 0$ for all $X\in \overline{\cK_1^*+\cK_2^*}$, and $\langle Y, Z\rangle <0$. Since $\cK_i^*\subseteq \overline{\cK_1^*+\cK^*_2}$, for $i=1, 2$, by definition $Y\in (\cK_i^*)^*$. Then by (iii) we have $Y\in \cK_i$. This means that $Y\in \cK_1\cap \cK_2$. On the other hand, $\langle Y, Z\rangle<0$, so $Z\notin (\cK_1\cap \cK_2)^*$. We conclude that the complement of $\overline{\cK_1^*+\cK_2^*}$ is a subset of the complement $(\cK_1\cap \cK_2)^*$, which means that $(\cK_1\cap \cK_2)^*\subseteq \overline{\cK_1^*+ \cK_2^*}$. 

For (vii) take the dual of both sides of (iv) and use (iii).
\section{Proof of Lemma~\ref{lem:embezzlement}}\label{app:embezzlement}

Let $\ket{\lambda_R} \in \mathbb C^{dR}$ be the state obtained from $\ket{\phi}\otimes \ket{\mu_R}$ via identification of $[d]\times[R]$ with $[dR]$, so the coordinates of $\ket{\lambda_R}$ are  $\frac{c_j} {\sqrt{r\chi_R}}$ for $(j, r)\in [d]\times [R]$. Let $\pi$ be the permutation that sorts these coordinates in the decreasing order. Therefore,
$$\ket{\tilde\lambda_R} =P_\pi \ket{\lambda_R} = \sum_{s=1}^R \tilde \lambda_s\ket s,$$
with the multi-set of $\tilde \lambda_s$'s being the same as the multi-set of $\frac{c_j} {\sqrt{r\chi_R}}$'s and $\tilde\lambda_1\geq \tilde\lambda_2\geq \cdots \geq \tilde\lambda_{dR}$. In the following we show that $\bra{\mu_R} \tilde \lambda_R\rangle$ is close to $1$ when $R$ is large. 

Let 
$$N_j^t = \Big|  \Big\{ r: \,       \frac{c_j} {\sqrt{r\chi_R}}> \frac{1}{\sqrt{t\chi_R}}     \Big \}   \Big|=\big| \big\{  r:\,  c_j^2 t> r     \big\}\big|.$$
Then we have $N_j^t< c_j^2t$. Moreover, by the normalization of $\ket{\phi}$ we have
$$\sum_{j=1}^d N_j^t< t\sum_{j=1}^d c_j^2 =t. $$ 
This means that 
$$\Big|\Big\{ s:\, \tilde \lambda_s>    \frac{1}{\sqrt{t\chi_R}}  \Big \}\Big|<t.$$
Next, using the fact that $\tilde \lambda_1\geq \cdots\geq \tilde \lambda_t$ we conclude that 
$$  \frac{1}{\sqrt{t\chi_R}}\geq \tilde \lambda_t, \qquad \forall t. $$ 
Therefore,
\begin{align*}
\bra{\mu_R} \tilde \lambda_R\rangle &= \sum_{r=1}^R \frac{\tilde \lambda_r}{\sqrt{r \chi_R}}
\geq \sum_{r=1}^R \tilde \lambda_r^2
\geq \Big(\sum_{j=1}^{d} c_j^2 \Big)\cdot \Big(  \sum_{s=1}^{\lfloor R/d\rfloor} \frac{1}{s\chi_R}   \Big) = \frac{\chi_{\lfloor R/d\rfloor}}{\chi_R}
\geq \frac{\ln R-\ln d}{\ln R+1}.
\end{align*}
Thus, $\bra{\mu_R} \tilde \lambda_R\rangle$ is arbitrarily close to $1$ for sufficiently large $R$. We are done since $\ket{\tilde \lambda_R}$ is obtained from $\ket\phi\otimes \ket{\mu_R}$ by applying a permutation.


\section{$\cK$ and $\cL$ are closed convex cones}\label{app:K-L-ccc}
In this appendix we show that the two sets $\cK, \cL$ defined in the proof of Lemma~\ref{lem:AlternativeSBMD_BipartiteSources} are closed convex cones. 

From the definition it is clear that $\cK$ is closed and that $rM\in \cK$ for $M\in \cK$ and $r\geq 0$. Then we need to show that for $P, Q\in \cK$, we have $M=P+Q\in \cK$. 
Let $\ket{v} = \sum_i v_i \ket{i}$ be an arbitrary vector. We compute 
\begin{align*}
\bra v \widehat M \ket v &\geq \sum_{i} |v_i|^2 M_{ii} - \sum_{i, j} |v_i v_j|\cdot |M_{ij}|  \\
& = \sum_{i, j} |v_i|^2 (P_{ii} + Q_{ii}) - \sum_{i\neq j} |v_iv_j|\cdot |P_{ij} + Q_{ij}| \\
& \geq \sum_{i, j} |v_i|^2 (P_{ii} + Q_{ii}) - \sum_{i\neq j} |v_iv_j|\cdot \big(|P_{ij}| + |Q_{ij}|\big) \\
& = \bra {\hat v} \widehat P\ket {\hat v} + \bra {\hat v} \widehat Q\ket {\hat v}\\
&\geq 0,
\end{align*} 
where $\ket {\hat v}$ is the vector whose coordinates are $|v_i|$'s, and in the last line we use $\widehat P, \widehat Q\succeq 0$. As a result, $\widehat M\succeq 0$ and $M\in \cK$. Therefore, $\cK$ is a closed convex cone. 

From the definition it is clear that $\cL$ is a convex cone. Thus we need to show that $\cL$ is closed. 
Indeed, the sum of closed convex cones is not necessarily closed, yet this holds for $\cL=\sum_\alpha \cL_\alpha$. To prove this, suppose that the sequence $\{X^{(j)}:j\geq 1\}\subset \cL$ tends to $X$ as $j\to \infty$. Since $X^{(j)}$'s and $X$ are positive semidefinite, for sufficiently large $j$ we have $0\preceq X^{(j)}\preceq X+I$. 
On the other hand, since $X^{(j)}$ belongs to $\cL$, there are $X^{(j)}_\alpha\in \cL_\alpha$ such that $X^{(j)}= \sum_\alpha X_\alpha^{(j)}$. Again using the fact that $X^{(j)}_\alpha$'s are positive semidefinite, for sufficiently large $j$ we have 
$$0\preceq X_\alpha^{(j)}\preceq X^{(j)}\preceq X+I.$$
Then by a compactness argument, there is subsequence $\{j_k:\, k\geq 1\}$ such that $\lim_{k\to \infty} X_\alpha^{(j_k)} =X_\alpha$ exists for all $\alpha$. Now since $\cL_\alpha$ is closed, we have $X_\alpha\in \cL_\alpha$. On the other hand, 
$$X=\lim_{k\to \infty} X^{(j_k)} = \lim_{k\to \infty} \sum_\alpha  X^{(j_k)}_\alpha = \sum_\alpha  X_\alpha.$$ 
This means that $X\in \cL$. Therefore, $\cL$ is also a closed convex cone.

\end{document}